\documentclass[12pt]{article}
\usepackage{amsfonts}
\usepackage{amsmath}
\usepackage{amssymb}
\usepackage{color}
\usepackage{eurosym}

\newtheorem{theorem}{Theorem}

\newtheorem{example}{Example}

\newenvironment{proof}[1][Proof]{\noindent\textbf{#1.} }{\ \rule{0.5em}{0.5em}}

\begin{document}

\title{Streaming problems as (multi-issue) claims problems\thanks{%
Financial support from grants PID2020-113440GBI00 and PID2023-146364NB-I00,
funded by MCIN/AEI/ 10.13039/501100011033 and
MICIU/AEI/10.13039/501100011033/ respectively, and by FEDER, and UE is
gratefully acknowledged.}}
\author{\textbf{Gustavo Berganti\~{n}os}\thanks{%
ECOBAS, Universidade de Vigo, ECOSOT, 36310 Vigo, Espa\~{n}a} \\
\textbf{Juan D. Moreno-Ternero}\thanks{%
Department of Economics, Universidad Pablo de Olavide, 41013 Sevilla, Espa%
\~{n}a; jdmoreno@upo.es}}
\maketitle

\begin{abstract}
We study the problem of allocating the revenues raised via paid
subscriptions to music streaming platforms among participating artists. We
show that the main methods to solve streaming problems (pro-rata,
user-centric and families generalizing them) can be seen as specific
(well-known) rules to solve (multi-issue) claims problems. Our results
permit to provide strong links between the well-established literature on
claims problems and the emerging literature on streaming problems.
\end{abstract}

\bigskip

\medskip {} 


\bigskip

\bigskip

\newpage

\bigskip

\bigskip

\section{Introduction}

The problem of adjudicating conflicting claims (in short, claims problem)
models a basic situation in which an endowment is allocated among agents who
have claims on it, and the available amount is not enough to fully honor all
claims. This is a classic problem that can be traced back to ancient
sources, such as Aristotle and the Talmud, although its formal treatment is
somewhat recent (e.g., O'Neill, 1982). A sizable literature at the
intersection of economics and operations research emerged afterwards (see,
for instance, Thomson (2019) and the literature cited therein). Some of the
efforts within that literature have been devoted to generalize claims
problems to account for multiple issues (e.g., Calleja et al., 2005; Ju et
al., 2007).\footnote{%
Acosta et al., (2022, 2023) have recently extended this multidimensional
setting to allow for crossed claims.} 

From a different vantage point, there has been a growing interest in recent
years to analyze platform businesses, which have transformed the ways in
which cultural content is produced and consumed (e.g., Aguiar et al., 2024). 
This is particularly the case with music. According to Statista, in the
third quarter of 2024, Spotify (the largest music streaming platform)
reached an all-time high with 640 million active users worldwide. 
The problem of sharing the revenue raised from paid subscriptions to
streaming platforms among artists, a new form of revenue sharing problems
under bundled pricing (e.g., Adams and Yellen, 1976; Ginsburgh and Zang,
2003; Berganti\~{n}os and Moreno-Ternero, 2015, 2016) is of critical
importance to understand the management of music platforms. We show in this
paper that such a problem can actually be approached resorting to the
literature on (multi-issue) claims problems mentioned above.%

Traditionally, music platforms used a \textit{pro-rata} method to award
artists based on the relative number of total streamings they achieved in
the platform. Gradually, they have been moving to a decentralized method in
which the amount paid by each user is allocated among the artists that user
streamed. If the latter allocation is proportional to individual streamings,
we obtain the so-called \textit{user-centric} method. If, instead, the
allocation does not distinguish among artists that were streamed, we obtain
the so-called \textit{Shapley} method. Varying the allocation method, we
obtain the family of \textit{probabilistic} methods (e.g., Berganti\~{n}os
and Moreno-Ternero, 2025). And compromising between the \textit{pro-rata}
and \textit{user-centric} methods we obtain the family of \textit{weighted}
methods (e.g., Berganti\~{n}os and Moreno-Ternero, 2025).

We show in this paper that all the methods described above can actually be
derived from existing rules to solve (multi-issue) claims problems. Key
among the latter are the so-called \textit{two-stage} rules, which first
divide the endowment among issues, and then divide the amount assigned to
each issue among the agents. More precisely, we show that the family of
probabilistic indices induces the same rewards as the family of two-stage
rules where we use the classical \textit{constrained equal awards} rule in the first
stage, and any claims rule (satisfying two basic requirements) for the
second stage. Similarly, we show that the family of user-weighted indices
induces the same rewards as the family of two-stage rules where we use the classical 
\textit{proportional} rule in the second stage, and any claims rule
(satisfying a basic requirement) for the first stage.

Our work can thus be interpreted as providing extension operators to convert
rules to solve claims problems into rules to solve streaming problems. As
such, we align with the literature dealing with operators on the space of
claims rules (e.g., Thomson and Yeh, 2008; Hougaard et al., 2012; 2013a;
2013b; Moreno-Ternero and Vidal-Puga, 2021). We also align with the
literature that applies the well-established framework from claims problems
to solve various related allocation problems with real-life implications. Recent instances are water conflicts (e.g., Mianabadi et al., 2014), linguistic subsidies (e.g., Ginsburgh and Moreno-Ternero, 2018),
broadcasting problems (e.g., Berganti\~{n}os and Moreno-Ternero, 2020), $%
CO_{2}$ emissions (e.g., Duro et al., 2020), slots allocation in sport competitions (e.g., Krumer and
Moreno-Ternero, 2020), or pollution abatement (e.g., Ju
et al., 2021).

The rest of the paper is organized as follows. In Section 2, we introduce
the basic aspects of claims problems. In Section 3, we introduce the basic
aspects of streaming problems. In Section 4, we provide the bridge between
both types of problems, as well as our main results. We conclude in Section
5.

\section{Claims problems}

Claims problems refer to an amount of a homogeneous and infinitely divisible
good (e.g., money) to be divided among a set of agents, who have claims on
the good that cannot be fully honored. See, for instance, Thomson (2019) for
an excellent recent survey on the topic.  A well-known example is bankruptcy. But many others
can also fit this framework. 
Formally, a \textbf{claims problem} is a triple $\left( N,c,E\right) $ where 
$N$ is the set of agents, $c\in \mathbb{R}_{+}^{N}$ is the claims vector and 
$E\in \mathbb{R}_{+}$ is the amount to be divided. It is assumed that $%
\sum\limits_{i\in N}c_{i}\geq E$. Let $\mathcal{C}$ denote the set of claims
problems.

A (claims) \textbf{rule} is a function $R$ assigning to each claims problem $%
\left( N,c,E\right)\in\mathcal{C} $ a vector $R\left( N,c,E\right) \in 
\mathbb{R}^{N}$ such that $\sum\limits_{i\in N}R_{i}\left( N,c,E\right) =E.$%
%
%

Some examples of well-known rules are presented next.

The \textbf{constrained equal awards rule} equalizes the amount received by
each agent as much as possible, subject to the constraint that no agent gets
an amount above his claim. Formally, for each $(N,c,E)\in\mathcal{C}$, and
each $i\in N$, 
\begin{equation*}
CEA_{i}(N,c,E)=\min \{\lambda ,c_{i}\},
\end{equation*}%
where $\lambda $ satisfies $\sum\limits_{i\in N}\min \{\lambda ,c_{i}\}=E$.

The \textbf{proportional rule} yields awards proportionally to claims.
Formally, for each $(N,c,E)\in\mathcal{C}$, and each $i\in N$, 
\begin{equation*}
P_{i}(N,c,E)=\frac{c_{i}}{\sum\limits_{j\in N}c_{j}}\,E.\text{ }
\end{equation*}

The \textbf{weighted proportional rule} yields awards proportionally to
weighted claims. Formally, given weights $\left(w_{i}\right) _{i\in N}$, for
each $(N,c,E)\in\mathcal{C}$ and each $i\in N,$ 
\begin{equation*}
P_{i}^{w}(N,c,E)=\frac{w_{i}c_{i}}{\sum\limits_{j\in N}w_{j}c_{j}}\,E.\text{ 
}
\end{equation*}


A \textbf{multi-issue claims problem} (e.g., Ju et al., 2007) is a tuple $%
\left( N,K,c,E\right) $ where $N$ is a set of agents, $K$ is a set of
issues, $c=(c_{ij})_{i\in N,j\in K}\ $where for all $i\in N$ and $j\in K,$ $%
c_{ij}\geq 0$ denotes the characteristic of agent $i$ on issue $j$, and $%
E\in \mathbb{R}_{+}$ is the amount of a homogeneous and infinitely divisible
good to be divided.\footnote{%
Calleja et al. (2005) introduce \textit{multi-issue allocation situations},
which are a particular case of multi-issue claims problems.} 
Let $\mathcal{MC}$ denote the set of multi-issue claims problems. For each $%
\left( N,K,c,E\right)\in \mathcal{MC}$ and each $j\in K$, let $%
c_{.j}=(c_{ij})_{i\in N}$ and $C^{j}=\sum\limits_{i\in N}c_{ij}.$

A (multi-issue claims) \textbf{rule} is a function $R$ assigning to each
multi-issue claims problem $\left( N,K,c,E\right)\in \mathcal{MC} $ a vector 
$R\left( N,K,c,E\right) \in \mathbb{R}^{N}.$

A well-known family of rules is defined as follows. First, a weight function is a function $\omega :\mathbb{R}_{+}^{K}\times \mathbb{R}%
_{+}\rightarrow \mathbb{R}_{+}^{K}$, which assigns a probability
distribution $\omega (x,y)$; namely, $0\leq w_{j}\left( x,y\right) \leq 1$
for all $j\in K$ and $\sum\limits_{j\in K}w_{j}\left( x,y\right) =1$.
The \textbf{(multi-issue) weighted proportional rule} associated with $w:%
\mathbb{R}_{++}^{K}\times \mathbb{R}_{++}\rightarrow \mathbb{R}^{K}$ assigns
for each problem $\left( N,K,c,E\right) $ and each $i\in N$, 
\begin{equation*}
P_{i}^{w}\left( N,K,c,E\right) =\sum_{j\in K}\frac{c_{ij}}{C^{j}}\omega
_{j}(\left( C^{j}\right) _{j\in K},E)E.
\end{equation*}


Alternatively, the literature has also considered two-stage rules, where in
the first stage the endowment is divided among the issues, and in the second
stage the amount assigned to each issue is divided among the agents. 
The final amount received by each agent is the sum over all issues.
Formally, let $\psi $ and $\left\{ \phi ^{j}\right\} _{j\in K}$ be some
claims rules. The\textbf{\ two-stage rule} $R^{\psi ,\phi }$ is the rule
obtained from the following two-stage procedure:

\begin{enumerate}
\item First stage. We consider the claims problem among the issues $%
(K,c^{K},E)$, where $c^{K}=(c_{j}^{K})_{j\in K}$ and for each $j\in K$, $%
c_{j}^{K}=\sum\limits_{i\in N}c_{ij}$ . We compute $\psi (K,c^{K},E).$

\item Second stage. For each $j\in K$, we consider the claims problem $%
(N,c_{.j},\psi _{j}\left( K,c^{K},E\right) )$. We compute $\phi
^{j}(N,c_{.j},\psi _{j}\left( K,c^{K},E\right) )$.
\end{enumerate}

For each $i\in N$ we define, 
\begin{equation*}
R_{i}^{\psi ,\phi }(N,K,C,E)=\sum_{j\in K}\phi _{i}^{j}\left( N,c_{.j},\psi
_{j}\left( K,c^{K},E\right) \right) .
\end{equation*}

Moreno-Ternero (2009) and Berganti\~{n}os et al. (2010) study the two-stage
rule where the proportional rule is used in both stages (namely, $\psi =\phi
^{j}=P$ for all $j\in K)$.\footnote{See also M\'{a}rmol and Hinojosa (2020).} Lorenzo-Freire et al. (2010) and Berganti\~{n}os
et al. (2011) study the two-stage rule where the constrained equal awards
rule is used in both stages (namely, $\psi =\phi ^{j}=CEA$ for all $j\in K).$
Berganti\~{n}os et al. (2018) study the two-stage rule where the constrained
equal awards rule is used in the first stage and the proportional rule in
the second stage (namely, $\psi =CEA$ and $\phi ^{j}=P$ for all $j\in K).$




\section{Streaming problems}

We now consider streaming problems, as introduced by Berganti\~{n}os and
Moreno-Ternero (2025). Let $\mathbb{N}$ represent the set of all potential
artists and $\mathbb{M}$ the set of all potential users (of music streaming
platforms). We can assume, without loss of generality, that both $\mathbb{N}$
and $\mathbb{M}$ are sufficiently large. In particular, we assume that $|%
\mathbb{M}|\geq 3$. Each specific platform involves a specific (finite) set
of artists $N\subset \mathbb{N}$ and a specific (finite) set of users $%
M\subset \mathbb{M}$, with cardinalities $n$ and $m$, respectively. For ease
of notation, we typically assume that $N=\left\{ 1,...,n\right\} $ and $%
M=\left\{ 1,...,m\right\} $. For each pair $i\in N,j\in M$, let $t_{ij}$
denote the times user $j$ played (via streaming) contents uploaded by artist 
$i$ in a platform (briefly, \textit{streams}), during a certain period of
time (e.g., month). In most of the platforms, playing a streaming unit will
be equivalent to playing a song (for at least 30 seconds). Let $t=\left(
t_{ij}\right) _{i\in N,j\in M}$ denote the corresponding matrix encompassing
all streams. We assume that for each $j\in M,$ $\sum\limits_{i\in N}t_{ij}>0$
(namely, each user has streamed some content).

A \textbf{streaming problem} is a triple $P=\left( N,M,t\right) $. Following
Berganti\~{n}os and Moreno-Ternero (2025), we normalize the amount paid by
each user to $1$. Thus, the amount to be divided among artists in a problem $%
\left( N,M,t\right) $ is just $m$, the number of users. The set of problems
so defined is denoted by $\mathcal{P}$.

For each $j\in M,$ we denote by $t^{-j}$ the matrix obtained from $t$ by
removing the column corresponding to user $j$. For each artist $i\in N$, $%
T_{i}\left( N,M,t\right) =\sum\limits_{j\in M}t_{ij}$, denotes the total
times $i$ was streamed. Likewise, for each user $j\in M$, $T^{j}\left(
N,M,t\right) =\sum\limits_{i\in N}t_{ij}$, denotes the total times $j$
streamed content. Notice that, by assumption, $T^{j}\left( N,M,t\right) >0$.

We define the set of fans of each artist as the set of users who have
streamed content from the artist at least once. Formally, for each $i\in N,$ 
$F_{i}\left( N,M,t\right) =\left\{ j\in M:t_{ij}>0\right\} $. Similarly, we
define the list of artists of a user as those from which the user has
streamed content at least once. Formally, for each $j\in M,$ $L^{j}\left(
N,M,t\right) =\left\{ i\in N:t_{ij}>0\right\} $. The profile of user $j$ is
defined as the streaming vector associated to such a user. Namely, $%
t_{.j}\left( N,M,t\right) =\left( t_{ij}\right) _{i\in N}$. When no
confusion arises we write $T_{i}$ instead of $T_{i}\left( N,M,t\right) $, $%
T^{j}$ instead of $T^{j}\left( N,M,t\right) ,$ $F_{i}$ instead of $%
F_{i}\left( N,M,t\right) ,$ $L^{j}$ instead of $L^{j}\left( N,M,t\right) ,$
and $t_{.j}$ instead of $t_{.j}\left( N,M,t\right) .$


A popularity \textbf{index} for streaming problems is a
mapping that measures the importance of each artist in each problem.
Formally, for each problem $\left( N,M,t\right) \in \mathcal{P}$, $I\left(
N,M,t\right) \in \mathbb{R}_{+}^{n}$ and, for each pair $i,j\in N$, $%
I_{i}\left( N,M,t\right) \geq I_{j}\left( N,M,t\right) $ if and only if $i$
is at least as important as $j$ at problem $\left( N,M,t\right) $. We assume
that $\sum\limits_{i\in N}I_{i}\left( N,M,t\right) >0$.

The reward received by each artist $i\in N$ from the revenues generated in
each problem ($m$ because the amount paid by each user has been normalized
to 1) is based on the importance of that artist in that problem. Formally, 
\begin{equation*}
R_{i}^{I}\left( N,M,t\right) =\frac{I_{i}\left( N,M,t\right) }{%
\sum\limits_{k\in N}I_{k}\left( N,M,t\right) }m.
\end{equation*}

Note that any positive linear transformation of a given index generates the
same allocation of rewards. Formally, for each $\lambda >0$ and each index $%
I $, $R^{\lambda I}\equiv R^{I}$. Thus, unless stated otherwise, we shall
slightly abuse language to identify an index with all its positive linear
transformations.



The index used by most platforms is the so-called \textbf{pro-rata} index,
which simply measures importance by the total number of streams. Formally,
for each problem $\left( N,M,t\right) \in \mathcal{P}$ and each artist $i\in
N,$ 
\begin{equation*}
P_{i}\left( N,M,t\right) =T_{i}=\sum_{j\in M}t_{ij}.
\end{equation*}

Thus, the amount received by each artist $i\in N$ under $P$ is

\begin{equation*}
R_{i}^{P}\left( N,M,t\right) =\frac{T_{i}}{\sum\limits_{k\in N}T_{k}}m.
\end{equation*}

Another index used in some platforms is the so-called \textbf{user-centric}
index. The amount paid by each user is divided among the artists streamed by
the user, proportionally to the total number of streams. Then, the index
assigns each artist the sum of the amounts obtained across users. Formally,
for each problem $\left( N,M,t\right) $ and each artist $i\in N,$ 
\begin{equation*}
U_{i}\left( N,M,t\right) =\sum_{j\in M}\frac{t_{ij}}{T^{j}}.
\end{equation*}%
It is easy to see that $R^{U}\left( N,M,t\right) =U\left( N,M,t\right) .$


The above two indices have been largely studied, for instance, in Alaei et al., (2022), Meyn et al. (2023), Gon\c{c}alves-Dosantos et al., (2024, 2025) and Berganti\~{n}os and Moreno-Ternero (2025).

We also consider the \textbf{Shapley} index, studied in Berganti\~{n}os and
Moreno-Ternero (2024) and Gon\c{c}alves-Dosantos et al., (2024).\footnote{%
The name is due to the fact that it corresponds with the Shapley value
(e.g., Shapley, 1953) of an associated cooperative game.} 
Formally, for each $\left( N,M,t\right) \in \mathcal{P}$ and each $i\in N$, 
\begin{equation*}
Sh_{i}\left( N,M,t\right) =\sum_{j\in M}\sum_{i\in L^{j}}\frac{1}{\left\vert
L^{j}\right\vert }.
\end{equation*}

Thus, under the Shapley index the subscription of each user $j$ is equally
allocated among the artists $j$ streamed. As in the case of user-centric,
the amount received by each artist from this index is precisely just what
the index indicates. Namely, 
\begin{equation*}
R^{Sh}\left( N,M,t\right) =Sh\left( N,M,t\right) .
\end{equation*}


We compare the three indices in the following example.

\begin{example}
\label{ex 2,3,a} Let $N=\left\{ 1,2,3\right\} ,$ $M=\left\{ 1,2\right\} ,$
and 
\begin{equation*}
t=\left( 
\begin{array}{cc}
10 & 0 \\ 
20 & 0 \\ 
0 & 70%
\end{array}%
\right) .
\end{equation*}


We then have the following: 
\begin{equation*}
\begin{tabular}{cccc}
& $R_{i}^{P}\left( N,M,t\right) $ & $R_{i}^{U}\left( N,M,t\right) $ & $%
R_{i}^{Sh}\left( N,M,t\right) $ \\ 
1 & 0.2 & 0.33 & 0.5 \\ 
2 & 0.4 & 0.66 & 0.5 \\ 
3 & 1.4 & 1 & 1%
\end{tabular}%
\end{equation*}
\end{example}

\bigskip


A weight system is a function $\omega :\mathbb{M}\times \mathbb{Z}%
_{+}^{n}\rightarrow \mathbb{R}$ such that for each $j\in \mathbb{M}$ and
each $x\in \mathbb{Z}_{+}^{n},$ $\omega \left( j,x\right) >0.$ For each
weight system $\omega $, its \textbf{weighted index} $I^{\omega }$ is
defined so that for each $\left( N,M,t\right) \in \mathcal{P}$ and each $%
i\in N,$ 
\begin{equation*}
I_{i}^{\omega }\left( N,M,t\right) =\sum_{j\in M}\omega \left(
j,t_{.j}\right) t_{ij}.
\end{equation*}

The value the index yields for each artist is obtained as the sum, over all
users, of the streams of the user-weighteded by a factor that depends on the
user and the streaming profile of the user. Both pro-rata and user-centric
are weighted indices. However, the Shapley index is not a weighted index
(e.g., Berganti\~{n}os and Moreno-Ternero, 2024; 2025).

A probability system is a function $\rho :\mathbb{M}\times \mathbb{Z}%
_{+}^{n}\rightarrow \mathbb{R}^{n}$ such that for each $j\in \mathbb{M}$ and
each $x\in \mathbb{Z}_{+}^{n},$ $0\leq \rho _{i}\left( j,x\right) \leq 1$, $%
\rho _{i}\left( j,x\right) =0$ when $x_{i}=0, $ and $\sum\limits_{i=1}^{n}%
\rho _{i}\left( j,x\right) =1.$ For each probability system $\rho $, its 
\textbf{probabilistic index }$I^{\rho }$ is defined so that for each $\left(
N,M,t\right) \in \mathcal{P}$ and each $i\in N,$ 
\begin{equation*}
I_{i}^{\rho }\left( N,M,t\right) =\sum_{j\in M}\rho _{i}\left(
j,t_{.j}\right) .
\end{equation*}

Then, the subscription of each user $j$ is fully allocated among the artists 
$j$ streamed (those that $j$ did not play get nothing) in an arbitrary way
that depends on $j$'s streaming vector. Both the user-centric and the
Shapley indices are probabilistic. However, the pro-rata index is not (e.g.,
Berganti\~{n}os and Moreno-Ternero, 2024; 2025). 


\section{A bridge between claims and streaming problems}

We could first associate with each streaming problem $\left( N,M,t\right)\in%
\mathcal{P}$ a claims problem in a trivial way. To wit, let the (overall)
amount paid by users be the amount to be divided. And let the claim of each
artist be its number of streams. Formally, we consider 
\begin{equation*}
\left( N,c\left( N,M,t\right) ,E\left( N,M,t\right) \right)\in\mathcal{C} ,
\end{equation*}%
where, for each $i\in N$, $c_{i}\left( N,M,t\right) =T_{i}$, and $E\left(
N,M,t\right) =m.$ As no (streaming) rule depending on $t_{ij}$
can be computed in this way, we  
associate instead with each streaming problem $\left( N,M,t\right)\in%
\mathcal{P} $ a multi-issue claims problem. Namely, 
\begin{equation*}
\left( N,K\left( N,M,t\right) ,c\left( N,M,t\right) ,E\left(
N,M,t\right)\right) \in\mathcal{MC} ,
\end{equation*}%
where $K\left( N,M,t\right) =M,$ $c\left( N,M,t\right) =t$, and $E\left(
N,M,t\right) =m.$ When no confusion arises, we write $\left( N,K,c,E\right) $
instead of $\left( N,K\left( N,M,t\right) ,c\left( N,M,t\right) ,E\left(
N,M,t\right) \right) .$

Notice that the set of issues is the set of users, the characteristic of
agent $i$ on issue $j$ is just the number of streams of user $j$ on artist $%
i,$ and the amount to be divided is the amount paid by the users (which
coincides with the number of issues because of the normalization). Thus, the
set of (multi-issue) claims problems associated with streaming problems
corresponds with the ones where $c_{ij}\in \mathbb{Z}_{+}^{n}$ for all $i\in
N$ and $j\in M,$ and $E=\left\vert K\right\vert .$


With the previous association, two-stage rules have a natural interpretation
in streaming problems. In the first stage, we decide the importance of each
user in the final allocation (taking into account the total number of
streams of each user). We allow, for instance, that more active users have
more importance. In the second stage, we decide, for each user, the
importance of each artist (taking into account the streams of the user).



\subsection{Results for individual indices}

In our first result, we provide the connections between the three basic
indices presented above and claims rules. More precisely, we show that the
rewards associated with the pro-rata and user-centric indices belong to the
family of multi-issue weigthed proportional rules, whereas this is not the
case for the Shapley index. We also show that the rewards associated with
the three indices can be seen as two-stage rules associated with the two
well-known (unidimensional) claims rules: the proportional rule and the
constrained equal awards rule.


\begin{theorem}
\label{th single rules} Let $\left( N,M,t\right)\in\mathcal{P} $ be a
streaming problem and $\left( N,K,c,E\right)\in\mathcal{MC} $ be the
associated (multi-issue) claims problem. Then, the following statements hold:

\begin{itemize}
\item There exists a weight function $w^{P}$ such that $R^{P}\left(
N,M,t\right)=P^{w^{P}}\left( N,K,c,E\right)$. 

\item There exists a weight function $w^{U}$ such that $U\left(
N,M,t\right)=P^{w^{U}}\left( N,K,c,E\right)$.

\item There does not exist a weight function $w^{S}$ such that $S\left(
N,M,t\right)=P^{w^{S}}\left( N,K,c,E\right)$.

\item $R^{P}\left( N,M,t\right) =R^{P,P}\left( N,K,c,E\right) .$

\item $U\left( N,M,t\right) =R^{CEA,P}\left( N,K,c,E\right) .$

\item $Sh\left( N,M,t\right) =R^{CEA,CEA}\left( N,K,c,E\right).$
\end{itemize}
\end{theorem}

\bigskip

\begin{proof}
For the first item, we consider $w^{P}:\mathbb{R}_{++}^{K}\times \mathbb{R}%
_{++}\rightarrow \mathbb{R}^{K}$ where for each $\left( \left( C^{j^{\prime
}}\right) _{j^{\prime }\in K},E\right) $, 
\begin{equation*}
\omega _{j}^{P}\left( \left( C^{j^{\prime }}\right) _{j^{\prime }\in
K},E\right) =\frac{C^{j}}{\sum\limits_{j^{\prime }\in K}C^{j^{\prime }}}%
\text{ for each }j\in M.
\end{equation*}

Then, for each $i\in N,$ 
\begin{equation*}
P_{i}^{\omega ^{P}}\left( N,K,c,E\right) =\sum_{j\in M}\frac{t_{ij}}{T^{j}}%
\frac{T^{j}}{\sum\limits_{j^{\prime }\in K}T^{j^{\prime }}}m=\frac{%
\sum\limits_{j\in M}t_{ij}}{\sum\limits_{j^{\prime }\in K}T^{j^{\prime }}}m=%
\frac{T_{i}}{\sum\limits_{i^{\prime }\in N}T_{i^{\prime }}}m=R_{i}^{P}\left(
N,M,t\right) .
\end{equation*}


For the second item, we consider $w^{U}:\mathbb{R}_{++}^{K}\times \mathbb{R}%
_{++}\rightarrow \mathbb{R}^{K}$ where for each $\left( \left( C^{j^{\prime
}}\right) _{j^{\prime }\in K},E\right) $, 
\begin{equation*}
\omega _{j}^{U}\left( \left( C^{j^{\prime }}\right) _{j^{\prime }\in
K},E\right) =\frac{1}{\left\vert K\right\vert }\text{ for each }j\in M.
\end{equation*}

Then, for each $i\in N,$ 
\begin{equation*}
P_{i}^{\omega ^{U}}\left( N,K,c,E\right) =\sum_{j\in M}\frac{t_{ij}}{T^{j}}%
\frac{1}{m}m=\sum_{j\in M}\frac{t_{ij}}{T^{j}}=U_{i}\left( N,M,t\right) .
\end{equation*}


For the third item, we note that each (multi-issue) weighted proportional
rule satisfies \textit{reallocation-proofnes} (e.g., Ju el al., 2007). That
is, for each $S\subset N$ such that $\sum\limits_{i\in
S}c_{ij}=\sum\limits_{i\in S}c_{ij}^{\prime },$ we have that $%
\sum\limits_{i\in S}R_{i}\left( N,K,c,E\right) =\sum\limits_{i\in
S}R_{i}\left( N,K,c^{\prime },E\right)$. Let $\left( N,M,t\right) $ and $%
\left( N,M,t^{\prime }\right) $ be such that $M=\left\{ j\right\} .$ The
associated claims problems are $\left( N,K,c,E\right) $ and $\left(
N,K,c^{\prime },E\right) $ where $K=\left\{ j\right\} $, $E=1$, $c=t$ and $%
c^{\prime }=t^{\prime }$. Assume that $N=\left\{ 1,2,3\right\} ,$ $t=\left(
20,0,10\right) $, $t^{\prime }=\left( 10,10,10\right) ,$ and $S=\left\{
1,2\right\} .$ Then, 
\begin{equation*}
\sum\limits_{i\in S}Sh_{i}\left( N,M,t\right) =\frac{1}{2}\text{ and }%
\sum\limits_{i\in S}Sh_{i}\left( N,M,t^{\prime }\right) =\frac{2}{3}.
\end{equation*}

Thus, $Sh$ does not satisfy reallocation proofness and, hence, it cannot be
associated to a (multi-issue) weighted proportional rule. 

For the fourth item, let $i\in N$ be given and compute $R_{i}^{P,P}\left(
N,K,c,E\right)$.

First stage. For each $j\in K,$ 
\begin{equation*}
P_{j}(K,c^{K},E)=\frac{T^{j}}{\sum\limits_{j^{\prime }\in K}T^{j^{\prime }}}%
m.
\end{equation*}

Second stage. 
\begin{equation*}
P_{i}(N,c_{.j},P_{j}\left( K,c^{K},E\right) )=\frac{t_{ij}}{T^{j}}\frac{T^{j}%
}{\sum\limits_{j^{\prime }\in M}T^{j^{\prime }}}m=\frac{t_{ij}}{%
\sum\limits_{j^{\prime }\in M}T^{j^{\prime }}}m.
\end{equation*}

Thus, for each $i\in N$, 
\begin{eqnarray*}
R_{i}^{P,P}(R,N,E,C) &=&\sum_{j\in K}P_{i}(N,c_{.j},P_{j}\left(
K,c^{K},E\right) ) \\
&=&\sum_{j\in M}\frac{t_{ij}}{\sum\limits_{j^{\prime }\in M}T^{j^{\prime }}}%
m=\frac{T_{i}}{\sum\limits_{i^{\prime }\in N}T_{i^{\prime }}}%
m=R_{i}^{P}\left( N,M,t\right) .
\end{eqnarray*}

For the fourth item, let $i\in N$ be given and compute $R_{i}^{CEA,P}\left(
N,K,c,E\right) .$

First stage. For each $j\in K,$ $CEA_{j}(K,c^{K},E)=\min \left\{ \lambda
,T^{j}\right\} $ where $\sum\limits_{j\in K}\min \left\{ \lambda
,T^{j}\right\} =E.$ It is straightforward to check that $\lambda =1.$ Hence, 
$CEA_{j}(K,c^{K},E)=1.$

Second stage. 
\begin{equation*}
P_{i}(N,c_{.j},CEA_{j}(K,c^{K},E))=\frac{t_{ij}}{T^{j}}.
\end{equation*}

Thus, for each $i\in N$, 
\begin{eqnarray*}
R_{i}^{CEA,P}(R,N,E,C) &=&\sum_{j\in K}P_{i}(N,c_{.j},CEA_{j}(K,c^{K},E))) \\
&=&\sum_{j\in M}\frac{t_{ij}}{T^{j}}=U_{i}\left( N,M,t\right) .
\end{eqnarray*}

For the sixth item, let $i\in N$ be given and compute $R_{i}^{CEA,CEA}%
\left(N,K,c,E\right)$.

First stage. Similarly to the previous item, we obtain that, for each $j\in
K $, $CEA_{j}(K,c^{K},E)=1.$

Second stage. 
\begin{equation*}
CEA_{i}(N,c_{.j},CEA_{j}(K,c^{K},E))=\left\{ 
\begin{tabular}{cc}
$\frac{1}{\left\vert L^{j}\right\vert }$ & if $i\in L^{j}$ \\ 
0 & otherwise%
\end{tabular}%
\right. .
\end{equation*}

Thus, for each $i\in N$, 
\begin{eqnarray*}
R_{i}^{CEA,CEA}(R,N,E,C) &=&\sum_{j\in
K}CEA_{i}(N,c_{.j},CEA_{j}(K,c^{K},E))) \\
&=&\sum_{j\in M}\sum_{i\in L^{j}}\frac{1}{\left\vert L^{j}\right\vert }%
=Sh_{i}\left( N,M,t\right) .
\end{eqnarray*}
\end{proof}

\bigskip

The last three items of Theorem $\ref{th single rules}$ allow us to consider
the three indices from a different perspective. The allocation rules they
induce (for streaming problems) can actually be described as two-stage
(claim) rules where we first decide the importance of each user and then the
importance of each artist for each user. The user-centric and Shapley
indices measure the importance of all users equally (applying the $CEA$
rule) whereas pro-rata measures the importance of each user proportionally
to the number of streams (applying the proportional rule). To measure the
importance of each artist for each user, pro-rata and user-centric do it
proportionally to the artists' streams (applying the proportional rule),
whereas the Shapley index states that all artists streamed by the user have
the same importance (applying the $CEA$ rule).

\bigskip

Theorem $\ref{th single rules}$ considers three possible two-stage rules
arising from combining the $CEA$ and $P$ rules in different ways. The other
combination between these two rules that is not captured in Theorem $\ref{th
single rules}$ is $R^{P,CEA}$, which we now illustrate. Let $\left(
N,M,t\right)\in\mathcal{P}$ be a streaming problem and $\left(
N,K,c,E\right)\in\mathcal{MC}$ be the associated (multi-issue) claims
problem. For each $i\in N$, 
\begin{eqnarray*}
R_{i}^{P,CEA}(R,N,E,C) &=&\sum_{j\in K}CEA_{i}(N,c_{.j},P_{j}\left(
K,c^{K},E\right) ) \\
&=&\sum_{j\in M}\min \left\{ \lambda _{j},t_{ij}\right\} .
\end{eqnarray*}


We illustrate this rule in the following example.

\begin{example}
\label{ex 3,2,a} Let $N=\left\{ 1,2\right\} ,$ $M=\left\{ 1,2,3\right\} ,$
and 
\begin{equation*}
t=\left( 
\begin{array}{ccc}
1 & 1 & 1 \\ 
1 & 1 & 95%
\end{array}%
\right) .
\end{equation*}

First stage. 
\begin{equation*}
P(K,c^{K},E)=(0.06,0.06,2.88).
\end{equation*}

Second Stage. Given $j\in \left\{ 1,2\right\} ,$ $\lambda _{j}=0.03$ and 
\begin{equation*}
CEA(N,c_{.j},P_{j}\left( K,c^{K},E\right) )=\left( 0.03,0.03\right) .
\end{equation*}

Given $j=3,$ $\lambda _{j}=1.88$ and 
\begin{equation*}
CEA(N,c_{.j},P_{j}\left( K,c^{K},E\right) )=\left( 1,1.88\right) .
\end{equation*}

Then, 
\begin{equation*}
R^{P,CEA}(R,N,E,C)=\left( 1.06,1.94\right) .
\end{equation*}
\end{example}



\subsection{Results for families of indices}

We now move from the three specific indices considered above to the two
families of indices encompassing pairs of them.

We consider first the probabilistic indices (encompassing the user-centric
and Shapley indices). As the next result states, they are connected to
two-stage rules where the constrained equal awards rule is used in the first
stage. 


To present the result, we need some notation first. Let $\rho $ be a
probabilistic system. For each $j\in M$ we can define the claims rule $\rho
^{j}$ induced by the probabilistic system $\rho $ as follows. For each $%
\left( N,c,E\right)\in\mathcal{C}$, and each $i\in N$, we define $\rho
_{i}^{j}\left( N,c,E\right) =\rho _{i}\left( j,c\right) E.$

We say that a claims rule $R$ satisfies \textit{non-negativity} and \textit{%
dummy} if for each $\left( N,c,E\right) \in \mathcal{C}$ and each $i\in N$, $%
R_{i}\left( N,c,E\right) \geq 0$ and $R_{i}\left( N,c,E\right) =0$ when $%
c_{i}=0.$ 

Let $\left( N,M,t\right)\in\mathcal{P} $ be a streaming problem and $\left(
N,K,c,E\right)\in\mathcal{MC} $ be the associated (multi-issue) claims
problem. We define the set of rewards induced by probabilistic indices as 
\begin{equation*}
RPI\left( N,M,t\right) =\left\{ x\in \mathbb{R}^{n}:x=I^{\rho }\left(
N,M,t\right) \text{ for some probabilistic index }I^{\rho }\right\} .
\end{equation*}

We define the set of allocations induced by two-stage rules $R^{\psi ,\phi }$
where $\psi =CEA$ and, for each $j\in K$, $\phi ^{j}$ is a claims rule that
satisfies non-negativity and dummy, as 
\begin{equation*}
CND\left( N,K,c,E\right) =\left\{ 
\begin{array}{c}
x\in \mathbb{R}^{n}:x=R^{CEA,\phi }\left( N,K,c,E\right) \text{ } \\ 
\text{where for each }j\in K,\text{ }\phi ^{j}\text{ satisfies
non-negativity and dummy}%
\end{array}%
\right\}
\end{equation*}

\begin{theorem}
\label{th prob index}Let $\left( N,M,t\right)\in\mathcal{P} $ be a streaming
problem and $\left( N,K,c,E\right)\in\mathcal{MC} $ be the associated
(multi-issue) claims problem. Then, the following statements hold:

\begin{itemize}
\item For each probabilistic index $I^{\rho }$, $I^{\rho }\left(
N,M,t\right) =$ $R^{CEA,\rho }\left( N,K,c,E\right) $ where $\rho =\left\{
\rho ^{j}\right\} _{j\in K}.$

\item $RPI\left( N,M,t\right) =CND\left( N,K,c,E\right) .$
\end{itemize}
\end{theorem}


\begin{proof}
For the first item, given $i\in N,$ we compute $R_{i}^{CEA,\rho }\left(
N,K,c,E\right) .$

First stage. We know from the proof of Theorem \ref{th single rules} that
for each $j\in K,$ $CEA_{j}(K,c^{K},E)=1.$

Second stage. 
\begin{equation*}
\rho _{i}^{j}(N,c_{.j},CEA_{j}(K,c^{K},E))=\rho _{i}\left( j,c_{.j}\right) .
\end{equation*}

Thus, for each $i\in N$, 
\begin{equation*}
R_{i}^{CEA,\rho }\left( N,K,c,E\right) =\sum_{j\in K}\rho _{i}\left(
j,c_{.j}\right) =\sum_{j\in M}\rho _{i}\left( j,t_{.j}\right) =I_{i}^{\rho
}\left( N,M,t\right) .
\end{equation*}%
This proves the first item. Note that it also proves one implication of the
second item, i.e., 
$RPI\left( N,M,t\right) \subseteq CND\left( N,K,c,E\right) $. 

We now prove the second implication, i.e., $CND\left( N,K,c,E\right)
\subseteq RPI\left( N,M,t\right) $. Let $x\in \mathbb{R}^{n}$ be such that $%
x=R^{CEA,\phi }\left( N,K,c,E\right) $ where for each $j\in K,$ $\phi ^{j}\ $%
satisfies non-negativity and dummy. We now define the probability system $%
\rho :\mathbb{M}\times \mathbb{Z}_{+}^{n}\rightarrow \mathbb{R}^{n}$ as
follows. For each $j\in \mathbb{M}$ and each $y\in \mathbb{Z}_{+}^{n}$, 
\begin{equation*}
\rho _{i}\left( j,y\right) =\phi _{i}^{j}\left( N,y,1\right) .
\end{equation*}%
As $\phi ^{j}$ is a claims rule that satisfies non negativity and dummy, we
have that $0\leq \rho _{i}\left( j,y\right) ,$ $\rho _{i}\left( j,y\right) =0
$ when $y_{i}=0,$ and $\sum\limits_{i=1}^{n}\rho _{i}\left( j,y\right) =1.$
Thus, $\rho $ is a probabilistic system. Now, 
\begin{eqnarray*}
I_{i}^{\rho }\left( N,M,t\right)  &=&\sum_{j\in M}\rho _{i}\left(
j,t_{.j}\right) =\sum_{j\in M}\phi _{i}^{j}\left( N,t_{.j},1\right)
=\sum_{j\in K}\phi _{i}^{j}\left( N,c_{.j},CEA_{j}(K,c^{K},E)\right)  \\
&=&R_{i}^{CEA,\phi }\left( N,K,c,E\right) =x_{i}.
\end{eqnarray*}
\end{proof}

\bigskip

The second item of Theorem \ref{th prob index} says that the family of
probabilistic indices induces the same rewards as the family of two-stage
rules with the constrained equal awards rule in the first stage (any
claims rule satisfying non-negativity and dummy can be used for the second
stage).

\bigskip

We now shift our focus to the family of weighted (rather than probabilistic)
indices (which encompass the pro-rata and user-centric indices).

We first consider the sub-family made of 
weighted indices where the weight depends only on the user (and not on the
streams). Formally, given a weighted index $I^{w}$ we say that $I^{w}$ is a 
\textbf{user-weighted index} if for each $j\in M$ and each pair $x,$ $%
x^{\prime }\in $ $\mathbb{Z}_{+}^{n},$ $w\left( j,x\right) =w\left(
j,x^{\prime }\right) .$ Thus, for each $j\in M$, we can define $%
w_{j}=w\left( j,x\right) $ with $x\in \mathbb{Z}_{+}^{n}.$

We show next that each user-weighted index is associated with a two-stage
rule where in the first stage we consider a weighted proportional rule and
in the second stage we consider the (unweighted) proportional rule. Besides,
we also show that the allocations induced by the rewards associated with
user-weighted indices coincide with the allocations induced by two-stage
rules where a rule satisfying positivity is used in the first stage, and the
proportional rule is used in the second stage.

Formally, we say that a claims rule $R$ satisfies \textit{positivity} if,
for each $\left( N,c,E\right)\in\mathcal{C}$ and each $i\in N$ such that $%
c_{i}>0$, $R_{i}\left( N,c,E\right) >0$.

Let $\left( N,M,t\right)\in\mathcal{P} $ be a streaming problem and $\left(
N,K,c,E\right)\in\mathcal{MC} $ be the associated (multi-issue) claims
problem. We define the set of rewards induced by user-weighted indices as 
\begin{equation*}
UWI\left( N,M,t\right) =\left\{ x\in \mathbb{R}^{n}:x=R^{I^{w}}\left(
N,M,t\right) \text{ for some user-weighted index }I^{w}\right\} .
\end{equation*}

We define the set of allocations induced by two-stage rules $R^{\psi ,\phi }$
where $\psi $ is a rule satisfying positivity and $\phi ^{j}=P$ for each $%
j\in M$ as 
\begin{equation*}
PP\left( N,K,c,E\right) =\left\{ 
\begin{array}{c}
x\in \mathbb{R}^{n}:x=R^{\psi ,\phi }\left( N,K,c,E\right) \text{ where }%
\psi \text{ is a rule satisfying positivity} \\ 
\text{and for each }j\in K,\text{ }\phi ^{j}=P%
\end{array}%
\right\}
\end{equation*}

\begin{theorem}
\label{th user weigh index}Let $\left( N,M,t\right)\in\mathcal{P} $ be a
streaming problem and $\left( N,K,c,E\right)\in\mathcal{MC} $ be the
associated (multi-issue) claims problem. Then, the following statements hold:

\begin{itemize}
\item For each user-weighted index $I^{w}$, $R^{I^{w}}\left( N,M,t\right) =$ 
$R^{\psi ,\phi }\left( N,K,c,E\right) $ where $\psi =P^{w}$ (the weighted
proportional rule given by $\left( w_{j}\right) _{j\in M})$ and $\phi ^{j}=P$
for each $j\in M.$

\item $UWI\left( N,M,t\right) =PP\left( N,K,c,E\right) .$
\end{itemize}
\end{theorem}


\begin{proof}
For the first item, given $i\in N,$ we compute $R_{i}^{\psi ,\phi }\left(
N,K,c,E\right) $ where $\psi =P^{w}$ and for each $j\in M,$ $\phi ^{j}=P.$

First stage. For each $j\in K,$ 
\begin{equation*}
P_{j}^{w}(K,c^{K},E)=\frac{w_{j}T^{j}}{\sum\limits_{j^{\prime }\in
K}w_{j^{\prime }}T^{j^{\prime }}}m.
\end{equation*}

Second stage. 
\begin{equation*}
P_{i}^{j}(N,c_{.j},P_{j}^{w}(K,c^{K},E))=\frac{t_{ij}}{T^{j}}\frac{w_{j}T^{j}%
}{\sum\limits_{j^{\prime }\in K}w_{j^{\prime }}T^{j^{\prime }}}m.
\end{equation*}

Thus, for each $i\in N$, 
\begin{equation*}
R_{i}^{\psi ,\phi }\left( N,K,c,E\right) =\sum_{j\in
K}P_{i}^{j}(N,c_{.j},P_{j}^{w}(K,c^{K},E))=\frac{\sum\limits_{j\in
M}w_{j}t_{ij}}{\sum\limits_{j\in M}w_{j}T^{j}}m.
\end{equation*}

Now, 
\begin{eqnarray*}
\sum\limits_{j\in M}w_{j}t_{ij} &=&\sum\limits_{j\in M}w\left(
j,t_{.j}\right) t_{ij},\text{ and} \\
\sum\limits_{j\in M}w_{j}T^{j} &=&\sum\limits_{j\in M}w_{j}\left( \sum_{i\in
N}t_{ij}\right) =\sum_{i\in N}\left( \sum\limits_{j\in M}w_{j}t_{ij}\right)
=\sum_{i\in N}\left( \sum\limits_{j\in M}w\left( j,t_{.j}\right)
t_{ij}\right).
\end{eqnarray*}

We thus deduce that $R_{i}^{\psi ,\phi }\left( N,K,c,E\right)
=R_{i}^{I^{w}}\left( N,M,t\right) $, which proves the first item, as well as
one implication of the second item, i.e., $UWI\left( N,M,t\right) \subseteq
PP\left( N,K,c,E\right) .$ 

We now prove the other implication, i.e., that $PP\left( N,K,c,E\right)
\subseteq UWI\left( N,M,t\right) $. Let $x\in \mathbb{R}^{n}$ be such that $%
x=R^{\psi ,\phi }\left( N,K,c,E\right) $ where $\psi $ is a rule satisfying
positivity, and for each $j\in K,$ $\phi ^{j}=P.$ We now define the weight
system $\omega :\mathbb{M}\times \mathbb{Z}_{+}^{n}\rightarrow \mathbb{R}$
as follows. Given $j\in \mathbb{M}$ and $y\in \mathbb{Z}_{+}^{n},$ we define 
\begin{equation*}
w\left( j,y\right) =\frac{\psi _{j}\left( M,\left( T^{j}\right) _{j\in
M},m\right) }{T^{j}}.
\end{equation*}

As $\psi $ is a rule satisfying positivity and $T^{j}>0$ for all $j\in M,$
we have that $w\left( j,y\right) >0.$ As $w$ does not depend on $y$, we
conclude that $w$ is a user-weighted system. 

Let $i\in N$. Then,

\begin{eqnarray*}
R_{i}^{I^{w}}\left( N,M,t\right) &=&\frac{\sum\limits_{j\in M}w\left(
j,t_{.j}\right) t_{ij}}{\sum\limits_{i\in N}\left( \sum\limits_{j\in
M}w\left( j,t_{.j}\right) t_{ij}\right) }m=\frac{\sum\limits_{j\in M}\frac{%
\psi _{j}\left( M,\left( T^{j}\right) _{j\in M},m\right) }{T^{j}}t_{ij}}{%
\sum\limits_{i\in N}\left( \sum\limits_{j\in M}\frac{\psi _{j}\left(
M,\left( T^{j}\right) _{j\in M},m\right) }{T^{j}}t_{ij}\right) }m \\
&=&\sum\limits_{j\in M}\frac{\frac{\psi _{j}\left( M,\left( T^{j}\right)
_{j\in M},m\right) }{T^{j}}t_{ij}}{\sum\limits_{j\in M}\frac{\psi _{j}\left(
M,\left( T^{j}\right) _{j\in M},m\right) }{T^{j}}\left( \sum\limits_{i\in
N}t_{ij}\right) }m=\sum\limits_{j\in M}\frac{\psi _{j}\left( M,\left(
T^{j}\right) _{j\in M},m\right) }{\sum\limits_{j\in M}\psi _{j}\left(
M,\left( T^{j}\right) _{j\in M},m\right) }\frac{t_{ij}}{T^{j}}m \\
&=&\sum\limits_{j\in M}\psi _{j}\left( M,\left( T^{j}\right) _{j\in
M},m\right) \frac{t_{ij}}{T^{j}}=\sum\limits_{j\in M}\psi _{j}\left(
K,c^{K},E\right) \frac{c_{ij}}{c_{j}^{K}} \\
&=&\sum\limits_{j\in K}P_{i}^{j}\left( N,c_{.j},\psi _{j}\left(
K,c^{K},E\right) \right) =R_{i}^{\psi ,\phi }\left( N,K,c,E\right) =x_{i}.
\end{eqnarray*}
\end{proof}

\bigskip

The second item of Theorem \ref{th user weigh index} says that the family of
user-weighted indices induces the same rewards as the family of two-stage
rules where we use the proportional rule in the second stage (any claims
rule satisfying positivity can be used for the first stage).

\bigskip

We conclude considering another subfamily of weighted indices, arising when
the weight depends on the user and the total number of streams, but not on
how the total streams are divided among artists. Formally, given a weighted
index $I^{w}$ we say that $I^{w}$ is a \textbf{total-streams-weighted index}
if for all $x,$ $x^{\prime }\in $ $\mathbb{Z}_{+}^{n}$ and all $j\in M$ such
that $\sum\limits_{i\in N}x_{i}=\sum\limits_{i\in N}x_{i}^{\prime }$, $%
w\left( j,x\right) =w\left( j,x^{\prime }\right) .$ 
We show that these indices can be seen as specific (multi-issue) weighted
proportional rules, rather than two-stage rules.


\begin{theorem}
\label{th total str weig index}Let $\left( N,M,t\right)\in\mathcal{P} $ be a
streaming problem and $\left( N,K,c,E\right)\in\mathcal{MC} $ be the
associated (multi-issue) claims problem. Then, for each
total-streams-weighted index $I^{\omega },$ 
\begin{equation*}
R^{I^{w}}\left( N,M,t\right) =P^{w}\left( N,K,c,E\right),
\end{equation*}
where 
there exists a mapping $w^{\ast }:\mathbb{M}\times \mathbb{Z}\rightarrow 
\mathbb{R}$ such that $w^{\ast }\left( j,x\right) >0$ for each $\left(
j,x\right) \in \mathbb{M}\times \mathbb{Z}$ and $w_{j}\left( \left(
C^{j}\right) _{j\in K},E\right) =\frac{w^{\ast }\left( j,C^{j}\right) C^{j}}{%
\sum\limits_{k\in K}w^{\ast }\left( k,C^{k}\right) C^{k}}$ for each $j\in M$.
\end{theorem}

\bigskip

\begin{proof}
Given $i\in N,$ 
\begin{eqnarray*}
R_{i}^{I^{w}}\left( N,M,t\right) &=&\frac{\sum\limits_{j\in M}\omega \left(
j,t_{.j}\right) t_{ij}}{\sum\limits_{i^{\prime }\in N}\sum\limits_{j\in
M}\omega \left( j,t_{.j}\right) t_{i^{\prime }j}}m=\frac{\sum\limits_{j\in
M}\omega \left( j,t_{.j}\right) t_{ij}}{\sum\limits_{j\in M}\omega \left(
j,t_{.j}\right) \sum\limits_{i^{\prime }\in N}t_{i^{\prime }k}}m \\
&=&\frac{\sum\limits_{j\in M}\omega \left( j,t_{.j}\right) t_{ij}}{%
\sum\limits_{j\in M}\omega \left( j,t_{.j}\right) T^{j}}m.
\end{eqnarray*}

We define $w^{\ast }:\mathbb{M}\times \mathbb{Z}\rightarrow \mathbb{R}$ as
follows. Given $\left( j,x\right) \in \mathbb{M}\times \mathbb{Z}$, let $%
t_{.j}^{\ast }$ be such that $t_{1j}^{\ast }=x$ and $t_{ij}^{\ast }=0$ when $%
i\neq 1.$ We then define $w^{\ast }\left( j,x\right) =w\left( j,t_{.j}^{\ast
}\right) .$
Let $i\in N$. Then, 
\begin{eqnarray*}
P_{i}^{w}\left( N,K,c,E\right) &=&\sum_{j\in K}\frac{c_{ij}}{C^{j}}\omega
_{j}(\left( C^{j}\right) _{j\in K},E)E=\sum_{j\in M}\frac{t_{ij}}{T^{j}}%
\omega _{j}(\left( T^{j}\right) _{j\in K},m)m \\
&=&\sum_{j\in M}\frac{t_{ij}}{T^{j}}\frac{\omega ^{\ast }\left(
j,T^{j}\right) T^{j}}{\sum\limits_{j\in M}\omega ^{\ast }\left(
j,T^{j}\right) T^{j}}m=\frac{\sum\limits_{j\in M}\omega ^{\ast }\left(
j,t_{.j}^{\ast }\right) t_{ij}}{\sum\limits_{j\in M}\omega ^{\ast }\left(
j,T^{j}\right) T^{j}}m \\
&=&\frac{\sum\limits_{j\in M}\omega \left( j,t_{.j}\right) t_{ij}}{%
\sum\limits_{j\in M}\omega \left( j,T^{j}\right) T^{j}}m,
\end{eqnarray*}%
where the last equality holds because for each $j\in M,$ $\sum\limits_{i\in
N}t_{ij}^{\ast }=\sum\limits_{i\in N}t_{ij}=T^{j}.$
\end{proof}


\section{Conclusions}

We have explored in this paper the connection between the emerging
literature on streaming problems and the well-established literature on
(multi-issue) claims problems. We have seen that the most well-known indices
to solve streaming problems, and families generalizing them, can be seen as
specific rules to solve (multi-issue) claims problems. Our work thus paves
the road to consider various known rules for claims problems to solve
streaming problems (in new ways that have not been considered so far).
Similarly, it opens the possibility to obtain new axiomatic
characterizations for the  pro-rata index (which can also be seen as the
two-stage proportional rule), the user-centric index (which can also be seen
as the CEA-proportional two-stage rule) or the Shapley index (which can also
be seen as the two-stage CEA rule) in streaming problems, inspired by
existing characterizations for claims rules. Finally, it also suggests to
consider well-known families of rules for claims problems (e.g., Berganti%
\~{n}os and Vidal-Puga, 2004; Moreno-Ternero, 2011; van den Brink and
Moreno-Ternero, 2017) to solve streaming problems.

We conclude mentioning that the game-theoretical approach is another
indirect way of solving streaming problems. There is a long tradition in economics and operations research of analyzing problems involving agents' cooperation with a game-theoretical approach.
Classical instances are airport problems (e.g., Littlechild and Owen, 1973), bankruptcy problems from the Talmud (e.g., Aumann and Maschler, 1985) or cost-allocation problems (e.g., Tijs and Driessen,
1986). More recent instances are the allocation of benefits of horizontal cooperation (e.g., Lozano et
al., 2013), the sharing of fish aggregating devices (e.g., Berganti\~{n}os et al., 2023) or even explainable artificial intelligence (e.g., Borgonovo et al., 2024). In our case here, this (game-theoretical) approach amounts to associate each streaming problem with a cooperative TU-game (rather than with a claims problems) in which the worth of each coalition of artists is determined by the
amount paid by users that are only streaming those artists. The standard option is to
take a pessimistic stance in which only users that exclusively streamed
artists in a given coalition will count for the worth (e.g., Berganti\~{n}os
and Moreno-Ternero, 2025). 
Alternatively, one could consider an optimistic stance, in which all users
that streamed some artist in the coalition count (e.g., Berganti\~{n}os and
Moreno-Ternero, 2024). It turns out that the Shapley value of both
(pessimistic and optimistic) games would coincide, giving rise to the
Shapley index considered in this paper (e.g., Berganti\~{n}os and
Moreno-Ternero, 2024). Schlicher et al. (2024) associate another cooperative
game to a streaming problem and study the core. 


\newpage


\end{document}